\documentclass[envcountsame]{llncs}
\usepackage{amssymb,amsmath}
\usepackage{graphicx}

\newcommand{\Nat}{\mathbb{N}}
\newcommand{\mon}[2]{Q_{#1}^{#2}}
\newcommand{\nnp}{\mathrm{nnp}}
\newcommand{\FO}{\mathrm{FO}}
\newcommand{\GROUP}{\mathrm{GROUP}}
\newcommand{\MOD}{\mathrm{MOD}}
\newcommand{\Arb}{\mathrm{arb}}
\newcommand{\MAJ}{\mathrm{MAJ}}
\newcommand{\IL}{\mathrm{IL}}
\newcommand{\IR}{\mathrm{IR}}

\newcommand{\NL}{\mathbf{N}}

\newcommand{\defs}{::=}
\newsavebox{\meinebox}

\begin{document}
%
\title{Non-definability of languages by generalized first-order formulas over ($\mathbb{N}$,+)}
%
%
\titlerunning{Non-definability of languages by generalised first-order formulas over ($\mathbb{N}$,+)}
\author{Andreas Krebs \inst{1} \and
A V Sreejith\inst{2}}

\institute{ Wilhelm-Schickard Institut, Universtit\"at T\"ubingen, Germany
\and Institute of Mathematical Sciences, Chennai, India}

\maketitle

\begin{abstract}
We consider first-order logic with monoidal quantifiers over words. We show that all languages with a neutral letter, definable using the 
addition predicate are also definable with the order predicate as the only numerical predicate. Let $\mathcal{S}$ be a subset of monoids.
Let $\mathcal{L_S}$ be the logic closed under quantification over the monoids in $\mathcal S$, and $\NL$ be the class of neutral letter languages.
Then we prove that 
$$\mathcal{L_S}[\textless,+] \cap \NL = \mathcal{L_S}[<] \cap \NL$$
Our result can be interpreted as the Crane Beach conjecture to hold for the logic $\mathcal{L_S}$[\textless,$+$].
As a consequence we get %
the result of Roy and Straubing that
FO$+$MOD[\textless,$+$] collapses to FO$+$MOD[\textless]. For cyclic groups, we answer an open question of Roy and Straubing, proving that MOD[\textless,$+$] collapses to MOD[\textless].
Our result also shows that multiplication as a numerical predicate is necessary for Barrington's theorem to hold and also to simulate
majority quantifiers.

All these results can be viewed as separation results for highly uniform circuit classes. For example we separate 
FO[\textless,$+$]-uniform CC$^0$ from FO[\textless,$+$]-uniform ACC$^0$.
\end{abstract}

%

\section{Introduction}
Consider a language with a ``neutral letter", i.e. a letter which can be inserted or deleted from any word in the language without 
changing its membership. The neutral letter concept has turned out to be useful for showing non-expressibility results. 
It had been used to establish super linear lower bounds for bounded-width branching
programs \cite{barr_superLinLBBWBP} and for the number of wires in circuit classes \cite{koucky_wiresVsGates}; it also led to results in communication complexity \cite{bddMultiPartyCommCompl}.
But mostly the concept is known in the context of the Crane Beach conjecture  proposed in \cite{barr_cbconj}.
There it was conjectured that first order logic with arbitrary numerical predicates (denoted as $\Arb$)  collapses to 
first order logic with only linear ordering in the presence of a neutral letter. The idea is that, in the presence of a neutral letter, formulas cannot rely on the precise location of input letters and hence numerical predicates will be of little use. Let $\NL$ denote the class of languages with neutral letters. 
Let $\mathcal{S}$ be a set of finite monoids  and $\mathcal{L_S}$ be the logic closed under quantification, where the quantifiers are Lindstr\"om quantifiers over some monoid from $\mathcal{S}$. Then the Crane Beach conjecture says that 
$$\mathcal{L_S}[\Arb] \cap \NL = \mathcal{L_S}[<] \cap \NL.$$
The conjecture was refuted by Barrington et. al. \cite{barr_cbconj}, where they showed that it does not hold for the logic $\FO[<,+,*]$, i.e. 
first order logic ($\mathcal{S}$ consists of only the existential quantifier) using addition and multiplication relation. 
In the same paper, the authors proved that the conjecture holds for various other logics. The Boolean closure of the $\Sigma_1$-fragment of $\FO[\Arb]$ satisfies the conjecture; that is $\mathcal{B}(\Sigma_1)[\Arb] \cap \NL = \mathcal{B}(\Sigma_1)[<] \cap \NL$. 
Lautemann, Tesson and Th\'erien \cite{ltt06} considered quantifiers which can count modulo a prime $p$ (called $\MOD_p$). They obtain that $\mathcal{B}(\Sigma_1^{0,p})[\Arb] \cap \NL = \mathcal{B}(\Sigma_1^{0,p})[<] \cap \NL$, which is equivalent to $\MOD_p[\Arb] \cap \NL = \MOD_p[<] \cap \NL$.

Benedikt and Libkin \cite{libkin_relIntrStruc} in the context of collapse results in database theory showed that first order logic
with only addition satisfies the Crane Beach conjecture. A different proof of the result can be found in \cite{barr_cbconj}. 
We generalize this result to arbitrary monoidal quantifiers. Let $\mathcal{S}$ be a set of finite monoids. Consider the logic
$\mathcal{L_S}$ where the quantifiers are Lindstr\"om quantifiers whose languages are word problems of monoids in $S$. Our main result
(Theorem \ref{thm_cbconj}) is that the Crane Beach conjecture hold for the logic $\mathcal{L_S}[<,+]$; that is
$$\mathcal{L_S}[<,+] \cap \NL = \mathcal{L_S}[<]\cap\NL.$$

If $S$ is an aperiodic monoid, then the theorem is equivalent to the result of Benedikt and Libkin. 
For solvable monoids Roy and Straubing \cite{roy_defGenFO} 
(used ideas of Benedikt and Libkin to) showed that in the presence of neutral letters $\FO+\MOD[<,+]$  collapse to
$\FO+\MOD[<]$. In their paper they raised the question: does $\MOD[<,+]$ satisfy the Crane Beach conjecture?
This can be answered by our main theorem.

Our results can also be viewed from the perspective of descriptive complexity of circuit classes. The books
\cite{immerman_book, vollmer_book} present the close connection between logics with monoid quantifiers and circuit classes. 
We know that the set of languages accepted by uniform-$AC^0$ circuits
are exactly those definable by first order logic using order,
addition and multiplication relations. Similarly CC$^0$ (constant depth, polynomial size circuits with MOD-gates)
corresponds to $\MOD[<,+,*]$, ACC$^0$ corresponds to $\FO+\MOD[<,+,*]$, TC$^0$ corresponds to $\MAJ[<,+,*]$, 
and NC$^1$ corresponds to $\GROUP[<,+,*]$ (The ``group quantifier'' evaluates
over a finite group). It is a well known result that AC$^0$ is separated from ACC$^0$ \cite{furst_parity}, but relationships between 
most other classes are open. For example, we do not know whether CC$^0$ is different from ACC$^0$. In fact we do not know whether
$\MOD_6[<,+,*]$ contains uniform-AC$^0$. This explains why the Crane Beach conjecture for prime modulo quantifiers \cite{ltt06}, using
arbitrary predicates, cannot be easily extended to composite modulo quantifiers.

We look at these separation questions from the descriptive complexity
perspective. As a first step, one can ask the
question of separating the logics without the multiplication
relation. That is, can one separate $\MOD[<,+]$ from
$\FO+\MOD[<,+]$? Is $\GROUP[<,+]$ different from $\FO+\MOD[<,+]$? 

Behle and Lange \cite{behle_foLessUniform} gave a notion of interpreting $\mathcal{L_S}[<,+]$ as highly uniform circuit classes.  Our results therefore can be summarized as: every $\FO[<,+]$ uniform constant depth polynomial size circuit with gates that compute a product in $\mathcal{S}$ and that recognizes a language with a neutral letter can be made $\FO[<]$-uniform.

As a consequence of our main theorem we are able to separate these uniform versions of circuit classes.
For example: The theorem states that $\MOD[<,+]$ definable languages with a neutral letter are also definable in $\MOD[<]$. Since $\MOD[<]$ cannot simulate the existential quantifiers \cite{str_cirBook} we have that $\FO[<,+]$ and $\MOD[<,+]$ are incomparable. In fact we show that no group quantifier can simulate existential quantifier if only addition is available.
This gives an alternate proof of the known result \cite{roy_defGenFO} that $\FO+\MOD_m[<,+]$ cannot count modulo a
prime $p$, where $p$ does not divide $m$. 
Another consequence is that the majority quantifier cannot be simulated by group quantifiers if multiplication is not available, thus
separating $\MAJ[<,+]$ from $\FO+\GROUP[<,+]$. Barrington's theorem \cite{barr_NC1} says that word problems over any finite group can be defined by the logic which 
uses only the $S_5$ group quantifier (the group whose elements are the set of all permutations over $5$ elements) if addition and multiplication predicates are available.
Our result shows multiplication is necessary for Barrington's theorem to hold. In other words $S_5$ cannot define word problems over
$S_6$ if only addition is available. 

Non expressibility results for various logics which uses addition and a variety of quantifiers have been considered earlier. Lynch \cite{lynch82} proved that $\FO[<,+]$ cannot count modulo any number. Nurmonen \cite{nur_countQuant} and Niwi\'nski and Stolboushkin \cite{niw_yEq2x} looked at logics with counting quantifiers equipped with numerical predicates of form $y=px$ and a linear ordering. Ruhl \cite{ruhl_FOUnc}, Schweikardt \cite{schwei_FOUnc},
Lautemann et.al. \cite{lautemann_logcfl}, Lange \cite{lange_majQnt} all showed the limited expressive power of addition in the presence of majority quantifiers. Behle, Krebs and Reifferscheid \cite{behle_regLangMajQnt,behle_nonSolvGrpsNotinMaj} proved that non-solvable groups are not definable in the two variable fragment of $MAJ[<]$.

For the purpose of proof we work over infinite strings which contain finite number of non-neutral letters. 
Our general proof strategy is similar to Benedikt and Libkin \cite{libkin_relIntrStruc} or Roy and Straubing \cite{roy_defGenFO} and consists of three main steps.
\begin{enumerate}
 \item Given a formula $\phi \in \mathcal{L_S}[<,+]$, we give an infinite set $\mathcal{D} \in \Nat$ and an ``active domain formula'' $\phi' \in \mathcal{L_S}[<,+]$ such that for all words $w$ whose non neutral positions belong to $\mathcal{D}$ we have 
 $w \vDash \phi \Leftrightarrow w \vDash \phi'$.
 Active domain formulas quantify only over non-neutral letter positions. Our major contribution (Theorem \ref{thm_acd4ls}) is  showing
 this step.  
 \item We give another infinite set $T \subseteq \mathcal{D}$ and an active domain formula $\psi \in \mathcal{L_S}[<]$ such that for all words $w$ whose non neutral positions belong to $T$ we have 
 $w \vDash \phi' \Leftrightarrow w \vDash \psi$.
 This step follows from an application of Ramsey theory (Theorem \ref{thm_ramsey}).
 \item All active domain formulas in $\mathcal{L_S}[<]$ accept languages with a neutral letter. This is an easy observation given by Lemma \ref{lem_acdneutr}.
\end{enumerate}
Finally using these three steps we prove our main theorem.

The main step is to build an active domain formula. Hence we need to show how to simulate a quantifier by an active domain formula.
In the case of $\FO[<,+]$, the quantifiers, considered as Lindstr\"om quantifiers, have a commutative and idempotent monoid. Hence
neither the order in which the quantifier runs over the positions of the word is important, nor does it matter if positions are queried
multiple times.
In Roy and Straubing this idea was extended in such a way that in the simulation of the $\MOD$ quantifier (again a commutative monoid),
every position is taken into account exactly once. In their construction while replacing a $\MOD$ quantifier they need to add
additional $\FO$ quantifiers and hence their construction only allows to replace a $\MOD[<,+]$ formula by an active domain $\FO+\MOD[<,+]$ formula.
In this paper, we construct a formula that takes every position into account exactly once and in
the correct order. Moreover we do not introduce any new quantifier, but use only the quantifier that is replaced.
This enables us to obtain the Crane Beach conjecture for logics whose quantifiers have a non-commutative monoid or are groups. For
example $\MOD[<,+]$, $\GROUP[<,+]$, and $\FO+\GROUP[<,+]$.

In contrast to previous work, we do not construct an equivalent active domain formula, but only a formula that is equivalent for certain domains. We show that it is in general sufficient to show this for one infinite domain. We also introduce a combinatorial structure called \emph{Sorting Tree} which can be of interest on its own. Yet another contribution is to use inverse elements of groups to merge two sorted lists of numbers. 


\noindent We present our main theorem and its corollaries in Section \ref{sec_result} followed by a section with the proof of Theorem
\ref{thm_acd4ls}. Our main contribution is 
Section \ref{sec_lemmaProof}. There we replace group quantifiers by its active domain version. %
 
\bigskip
\section{Preliminaries} \label{sec_prelims}
An alphabet $\Sigma$ is a finite set of symbols. The set of all finite words over $\Sigma$ is denoted by $\Sigma^*$, the set of all right infinite words is denoted by $\Sigma^{\omega}$. Let $\Sigma^{\infty} = \Sigma^* \cup \Sigma^{\omega}$.  Consider a language $L \subseteq \Sigma^{\infty}$ and a letter $\lambda \in \Sigma$. We say that $\lambda$ is a \emph{neutral letter} for $L$ if for all $u,v \in \Sigma^{\infty}$ we have that $u\lambda v \in L \Leftrightarrow uv \in L$. We denote the set of all languages with a neutral letter by $\NL$. 

For a word $w\in\Sigma^\infty$ 
the notation $w(i)$ denotes the $i^{th}$ letter in $w$, i.e. $w=w(0)w(1)w(2)\dots$.
For a word $w$ in a language $L$ with neutral letter $\lambda$, we define the non-neutral positions $\nnp(w)$ of $w$ to be the set of all positions which do not have the neutral letter. 

A monoid is a set closed under a binary associative operation and has an identity element. All monoids we consider except for $\Sigma^*$ and $\Sigma^\infty$ will be finite. A monoid $M$ and $S \subseteq M$ defines a \emph{word problem}. Its language is composed of words $w \in M^*$, such that when the elements of $w$ are multiplied in order we get an element in $S$. 
We say that a monoid $M$ \emph{divides} a monoid $N$ if there exists a submonoid $N'$ of $N$ and a surjective morphism from $N'$ to
$M$. A monoid $M$ \emph{recognizes} a language $L\subseteq \Sigma^*$ if there exists a morphism $h:\Sigma^* \rightarrow M$ and a subset
$T \subseteq M$ such that $L = h^{-1}(T)$. It is known that finite monoids recognize exactly regular languages \cite{str_cirBook}. We
denote by $\mathcal{M}$ the set of all finite monoids, $\mathcal{G \subset M}$ the set of all finite groups and $\MOD$ the set of all
finite cyclic group. We denote by $U_1$ the monoid consisting of elements $\{0,1\}$ under multiplication. For a monoid $M$, the element
$1 \in M$ will denote its identity element. We also use the block product of monoids, whose definition can be found in
\cite{str_cirBook}. For a set $S$ of monoids, $bpc(S)$ denotes the smallest set which contains $S$ and is closed under block products. 


Given a formula $\phi$ with free variables $x_1,\dots,x_k$, we write $w,i_1,\dots,i_k\models\phi$ if $w$ is a model for the formula $\phi$ when the free variables $x_j$ is assigned to $i_j$ for $j=1,\dots,k$.
We abuse notation and let $c \in \Sigma$ also be the unary predicate symbols of the logic we consider. That is $w,i\models c(x)$ iff $w(i)=c$. 
Let $\mathcal{V}$ be a set of variables, $\mathcal{R}$ be a set of numerical predicates and $\mathcal{S \subseteq M}$. We define the
logic $\mathcal{L_S[\mathcal R]}$ to be built from the unary predicate symbols $c$, where $c \in \Sigma$, the binary predicate $\{=\}$,
the predicates in $\mathcal{R}$, the variable symbols $\mathcal{V}$, the Boolean connectives $\{\neg,\vee,\wedge\}$, and the monoid
quantifiers $\mon{M}{m}$, where $M \in \mathcal{S}$ is a monoid and $m \in M$. We also identify the logic class $\mathcal{L_S[\mathcal
R]}$ with the set of all languages definable in it.

Our definition of monoid quantifiers is a special case of Lindstr\"om quantifiers \cite{lin_genQuant}. The formal definition of a
monoid quantifier \cite{barr_uniformNC1} is as follows. Let $M=\{m_1,\dots,m_{K},1\}$ be a monoid with $K+1$ elements. For an $m \in
M$, the quantifier $\mon{M}{m}$ is applied on $K$ formulas. Let $x$ be a free variable and $\phi_1(x),\dots,\phi_{K}(x)$ be  $K$ formulas. Then 
$w \models \mon{M}{m} x \langle \phi_1(x),\dots,\phi_{K}(x) \rangle$
iff the word $u$ when multiplied gives the element $m$, i.e.\  $\prod_i u(i) = m$, where the $i^{th}$ letter of $u$, $0 \leq i < |w|$, is 
$$u(i) = \left\{ \begin{array}{rl}
 m_1 & \mbox{ if } w,i \models \phi_1\\
 m_2 & \mbox{ if } w,i \models \neg \phi_1 \wedge \phi_2\\
     & \vdots	\\
 m_{K} & \mbox{ if } w,i \models \neg \phi_1 \wedge \dots \wedge \neg \phi_{K-1} \wedge \phi_{K}\\
 1 & \mbox{ otherwise}
\end{array} \right.$$

The following ``shorthand'' notation is used to avoid clutter. We denote by $\mon{M}{m} x ~\phi~ \langle\alpha_1,\dots,\alpha_K\rangle$, the formula
$\mon{M}{m} x\langle \phi \wedge \alpha_1,\dots,\phi \wedge \alpha_K \rangle$. Informally, this relativizes the quantifier to the positions where $\phi$ is true, by multiplying the neutral element in all other places.

Consider the monoid $U_1$. It is easy to see that the word problem defined by $U_1$ and the set $\{0\}$ defines the regular language
$1^*0(0+1)^*$. Then $\mon{U_1}{0}$ is same as the existential quantifier $\exists$, since any formula of the form $\exists x \phi$ is equivalent
to $\mon{U_1}{0} x ~\langle \phi \rangle$. So the logic $\mathcal{L}_{U_1}[<]$ denotes first-order logic, $\FO[<]$. Let
$C_q$ stand for the cyclic group with $q$ elements. Then the quantifiers $\mon{C_q}{1}$ corresponds to modulo quantifiers
\cite{str_regGenQnt}. Thus $\mathcal{L}_{\MOD}[<]$ corresponds to all regular languages whose syntactic monoids are solvable groups
\cite{str_cirBook}. 
For a sentence $\phi \in \mathcal{L_S}[\mathcal R]$ we define $L(\phi) = \{w \mid w \vDash \phi\}$.
The following result gives an algebraic characterization for the logic $\mathcal{L_S}[<]$.
\begin{lemma}[\cite{str_cirBook}]
\label{lem_logicalgb}
Let $\mathcal{S \subseteq M}$. 
Let $L \subseteq \Sigma^*$ such that $M$ is the smallest monoid which recognizes $L$. 
Then $L$ is definable in $\mathcal{L_S}[<]$ iff $M$ divides a monoid in $bpc(S)$.
\end{lemma}
%
%
%
%
\bigskip
\section{Results}\label{sec_result} 
Let $\mathcal{S \subseteq M}$ be any set of monoids. 
We show that the Crane Beach conjecture is true for the logic $\mathcal{L_S}[<,+]$.
\begin{theorem}[Main Theorem]
\label{thm_cbconj}
 Let $\mathcal{S \subseteq M}$. Then $$\mathcal{L_S}[<,+] \cap \NL = \mathcal{L}_{\mathcal{S}}[<] \cap \NL$$
\end{theorem}
The proof of this theorem is given in Section \ref{sec_cranebeach}.

%

\subsection{Non definability Results}

Theorem \ref{thm_cbconj} give us the following corollaries. 
\begin{corollary} 
\label{cor_regLang}
All languages with a neutral letter in $\mathcal{L_M}[<,+]$ are regular.
\end{corollary}
\begin{proof}
 By Theorem \ref{thm_cbconj} we know that all languages with a neutral letter in $\mathcal{L_M}[<,+]$ can be defined in $\mathcal{L_M}[<]$ which by Lemma \ref{lem_logicalgb} is the set of all regular languages.
\qed \end{proof}

Recall that a monoid $M$ \emph{divides} a monoid $N$ if $M$ is a morphic image of a submonoid of $N$.

\begin{corollary}
  \label{cor_simGrp}
Let $\mathcal{S \subseteq G}$. Let $G$ be a simple group that does not divide any monoid $M$ in $\mathcal{S}$. Then the word problem over $G$ is not definable in $\mathcal{L_S}[<,+]$.
\end{corollary}
\begin{proof}
The word problem over $G$ has a neutral letter. The result now follows from Theorem \ref{thm_cbconj} and Lemma \ref{lem_logicalgb}.
\qed \end{proof}


The majority quantifier, $\mathrm{Maj} ~x ~\phi(x)$ is given as follows.
$$w \vDash \mathrm{Maj} ~x ~\phi(x) \Leftrightarrow |\{i \mid w \vDash \phi(i), ~i \leq |w| \}|>\frac{|w|}{2}$$
$\MAJ[<]$  denotes the logic closed under majority quantifiers.
It is known that the majority quantifier can be simulated by the non-solvable group $S_5$ if both multiplication and addition are available \cite{vollmer_book}. We show that multiplication is necessary to simulate majority quantifiers.
\begin{corollary}
$\MAJ[<] \nsubseteq \mathcal{L_M}[<,+]$.
\end{corollary}
\begin{proof}
 Consider the language $L \subseteq \{a,b,c\}^*$ consisting of all words with an equal number of $a$'s and $b$'s. $L$ can be proven to be definable in $\MAJ[<]$. Also note that $c$ is a neutral element for $L$. By Corollary \ref{cor_regLang}, and the fact that $L$ is nonregular,  we know that $L$ is not definable in $\mathcal{L_M}[<,+]$.
\qed \end{proof}

Barrington's theorem \cite{barr_NC1} says that the word problem of any finite group can be defined in the logic $\mathcal{L}_{S_5}[<,+,*]$. The following theorem shows that multiplication is necessary for Barrington's theorem to hold.
\begin{corollary}
  The word problem over the group $S_6$ is not definable in $\mathcal{L}_{S_5}[<,+]$. Infact there does not exist any one finite monoid $M$
  such that all regular languages can be defined in $\mathcal{L}_M[<,+]$.
\end{corollary}
\begin{proof}
  $A_6$ is a simple subgroup of $S_6$, which does not divide $S_5$. From Corollary \ref{cor_simGrp} it follows that the word problem over
  $S_6$ is not definable in $\mathcal{L}_{S_5}[<,+]$. \\
  For any finite monoid $M$, there exists a simple group $G$ such that $G$ does not divide $M$ and hence the word problem over $G$ is
  not definable in $\mathcal{L}_M[<,+]$.
\qed \end{proof}

Let $L_p$ be the set of all words $w \in \{0,1\}^*$ such that the number of occurrences of $1$ in $w$ is equal to $0 \pmod p$. Then we
get the result in \cite{roy_defGenFO} that $L_p$ is not definable in $\FO+\MOD_m[<,+]$, if $p$ is a prime which does not divide
$m$.
\begin{corollary}[\cite{roy_defGenFO}]
  If $p$ is a prime which does not divide $m$, then $L_p$ is not definable in $\FO+\MOD_m[<,+]$.
\end{corollary}
\begin{proof}
  Let $L_p$ be definable in $\FO+\MOD_m[<,+]$. Since $0$ is a neutral letter in $L_p$, Theorem \ref{thm_cbconj} says
  $L_p$ is also definable in $\FO+\MOD_m[<]$.  Due to Lemma \ref{lem_logicalgb} and \cite{str_cirBook}, this is a contradiction.
\qed \end{proof}

It is an open conjecture whether the language $1^*$ can be accepted by the circuit complexity class CC$^0$ \cite{str_cirBook}. It is also known that languages accepted by CC$^0$ circuits are exactly those which are definable by $\mathcal{L}_{\MOD}[<,+,*]$ formulas \cite{vollmer_book}. 

To progress in this direction Roy and Straubing \cite{roy_defGenFO} had posed the question of whether $1^* \notin \mathcal{L}_{\MOD}[<,+]$. Below we show that this is the case. 

\begin{corollary}
 $1^* \notin \mathcal{L}_{\MOD}[<,+]$. In fact $1^* \notin \mathcal{L_{G}}[<,+]$.
\end{corollary}
\begin{proof}
  The minimal monoid which can accept $1^*$ is $U_1$ and clearly the language is in $\NL$. By Theorem \ref{thm_cbconj} if
  there is a formula in $\mathcal{L_G}[<,+]$ which can define $1^*$, then $\mathcal{L_G}[<]$ can also define $1^*$.
  From Lemma \ref{lem_logicalgb} it follows that the monoid $U_1$ divides a group. But this is a contradiction \cite{str_cirBook}.
\qed \end{proof}

Behle and Lange \cite{behle_foLessUniform} give a notion of interpreting $\mathcal{L_S}[<,+]$ as highly uniform circuit classes.  As a consequence we can interpret the following results as a separation of the corresponding circuit classes.
\begin{corollary}
The following separation results hold, for all $m>1$
\begin{itemize}
\item $\FO[<,+] \not \subseteq \MOD[<,+]$.
\item $\MOD_m[<,+] \not \subseteq \FO[<,+]$.
\item $\FO[<,+] \subsetneq \FO+\MOD_m[<,+] \subsetneq \FO+\MOD[<,+]$
\item $\FO+\MOD[<,+] \subsetneq \FO+\GROUP[<,+]$
\item $\MAJ[<,+] \not \subseteq \FO+\GROUP[<,+]$
\end{itemize}
\end{corollary}


\subsection{Regular languages in $L_{\mathcal S}[<,+]$}
We now look at regular languages definable by the logic $L_{\mathcal{S}}[<,+]$, for an $\mathcal{S \subseteq M}$. We first show that this logic is closed under quotienting.

\begin{lemma} \label{lem_qnt}
Let $S \subseteq \mathcal{M}$ and $\Sigma$ be a finite alphabet. Let $L \subseteq \Sigma^*$ be definable in $L_{\mathcal{S}}[<,+]$ and $u, v \in \Sigma^*$. Then $u^{-1}Lv^{-1}$ is also definable in $L_{\mathcal{S}}[<,+]$.
\end{lemma}
\begin{proof}
\qed \end{proof}

We now show that the logic is also closed under inverse length perserving morphisms.
\begin{lemma} \label{lem_morph}
Let $S \subseteq \mathcal{M}$. Let $\Sigma,\Gamma$ be finite alphabets and let $h:\Gamma^* \rightarrow \Sigma^*$ be a homomorphism such that $h(\Gamma) \subseteq \Sigma^r$ for some fixed $r>0$. If $L \subseteq \Sigma^*$ is definable in $L_{\mathcal{S}}[<,+]$, then $h^{-1}(L) \subseteq \Gamma^*$ is also definable in $L_{\mathcal{S}}[<,+]$.
\end{lemma}
\begin{proof}
\qed \end{proof}

We now give an algebraic characterization for regular languages definable by $L_{\mathcal S}[<,+]$. 
\begin{theorem}
Let $\mathcal{S \subseteq M}$. Let $L \subseteq \Sigma^*$ be a regular language. Then $L$ is definable in $L_{\mathcal{S}}[<,+]$ iff there exists a semigroup $\mathcal{V}$ and a morphism, $h: \Sigma^* \rightarrow \mathcal{V}$, such that for all $k \in \Nat$, every monoid in $h(\Sigma^k)$ is also in $bpc(S)$.
\end{theorem}
\begin{proof}

\qed \end{proof}

Let $\mathcal{S}$ be a set of monoids such that, given a monoid $M$, it is decidable if $M$ divides a block product of monoids in $\mathcal{S}$. 
Then, given a regular language $L$, it is decidable if $L \in L_{\mathcal{S}}[<]$.
Together with our main theorem we get that it is decidable if $L\in\mathcal{L_{S}}[<,+]$.

\begin{corollary}
Let $\mathcal{S}$ be a set of monoids such that, given a monoid $M$, it is decidable if $M$ divides a block product of monoids in
$\mathcal{S}$. Then, given a regular language $L$, it is decidable if $L \in L_{\mathcal{S}}[<,+]$.
\end{corollary}

For $\FO+\MOD[<,+]$ this was proved in \cite{roy_defGenFO}. Here we prove this for the special case when $\mathcal{S}=\MOD$.
\begin{corollary}
Given a regular language $L$, the question whether $L$ is definable in $\MOD[<,+]$ is decidable.
\end{corollary}

\bigskip
\section{Proof of the Main Theorem}\label{sec_cranebeach}
In this section we handle the general proof steps as in Libkin or Roy and Straubing of removing the plus predicate from the formula in the presence of a neutral letter. We show that all these results go through even in the presence of general Lindstr\"om quantifiers. The new crucial step is Lemma \ref{lem_acd4grpqnt} where we convert a group quantifier to an active domain formula without introducing any other quantifiers. The proof of this lemma is deferred to the next section.

Let $\mathcal{S \subseteq M}$ be any nonempty set. To prove Theorem \ref{thm_cbconj} we will consider the more general logic, $\mathcal{L_S}[<,+,0,\{\equiv_q:q>1\}]$ over 
the alphabet $\Sigma$. In this logic $+$ is a binary function, $0$ is a constant, and $a \equiv_q b$ means $q$ divides $b-a$. 
The reason for introducing these new relations (which are definable using $+$) is to use a quantifier elimination procedure.
All languages recognized by this logic are in $\mathcal{L_S}[<,+]$. 

The formulas we consider will usually define languages with a neutral letter. 
Let an \emph{active domain formula} over a letter $\lambda \in \Sigma$ be a formula where all quantifiers are of the form:
$\mon{M}{m} x ~\neg \lambda(x) \langle \phi_1,\dots,\phi_K \rangle$.
That is the quantifiers, quantify only over the ``active domain'', the positions which does not contain the letter $\lambda$.
For the purpose of the proof we assume that the neutral letter language defined by a formula $\phi \in \mathcal{L_S}[<,+]$ is a subset of $\Sigma^*\lambda^{\omega}$. 
The idea is to work with infinite words, where the arguments are easier, since the variable range is not bounded by the word length.

For $r\in\mathbb N$ we define the set $\mathcal{D}_r = \{r^i \mid 0<i \in \Nat\}$. We say that a formula $\phi(x_1,\dots,x_t) \in
\mathcal{L_S}[<,+]$ \emph{collapses} to $\phi'$, if $\phi'$ is an active domain formula in $\mathcal{L_S}[<,+]$ and there exists an 
$\mathcal{R}_{\phi} \in \Nat$ such that for all $r \geq \mathcal{R}_{\phi}$, $w\in \Sigma^*\lambda^{\omega}$ with 
$\nnp(w) \subseteq \mathcal{D}_r$ and for all $a_1,\dots,a_t \in \Nat$ we have that 
$$w \models \phi(a_1,\dots,a_t) \Leftrightarrow w \models \phi'(a_1,\dots,a_t)$$ 
In the above definition we say that $\mathcal{R}_{\phi}$ collapse $\phi$ to $\phi'$.

The results by Benedikt and Libkin \cite{libkin_relIntrStruc}, and Roy and Straubing \cite{roy_defGenFO} show that for all formulas $\phi \in \mathcal{L}_{\MOD \cup U_1}[<,+]$ 
there exists an active domain formula $\phi'$ in that logic, such that for all words $w \in \Sigma^*\lambda^{\omega}$, $w \vDash \phi \Leftrightarrow w \vDash \phi'$. 
They assume no restriction on the non-neutral positions of $w$. Observe that our collapse result is different from theirs. We prove that if we consider only words, 
whose non-neutral positions are in $\mathcal{D}_r$, then any formula $\phi \in \mathcal{L_S}[<,+]$ is equivalent to the active domain
formula $\phi' \in \mathcal{L_S}[<,+]$. 
That is, we are not concerned about the satisfiability of those words with non-neutral positions not in $\mathcal{D}_r$.

We show that formulas with a group quantifier, $G \in \mathcal{S}$ can be collapsed.
\begin{lemma}
 \label{lem_acd4grpqnt}
Let $\phi = \mon{G}{m} z \langle \phi_1,\dots,\phi_K \rangle$ be in $\mathcal{L_S}[<,+]$. Assume formulas\linebreak $\phi_1,\dots,\phi_K$
collapse. Then $\phi$ collapses to an active domain formula $\phi'$.
\end{lemma}

The proof of Lemma \ref{lem_acd4grpqnt} will be given in Section \ref{sec_lemmaProof}. 
Benedikt and Libkin \cite{libkin_relIntrStruc} gives a similar theorem for the monoid $U_1$ (the existential quantifier).
\begin{lemma}[\cite{libkin_relIntrStruc}]
 \label{lem_acd4u1}
Let $\phi = \mon{U_1}{m} z \langle \phi_1,\dots,\phi_K \rangle$ be a formula in $\mathcal{L_S}[<,+]$. Let us assume that formulas
$\phi_1,\dots \phi_K$ collapse. Then $\phi$ collapses to an active domain formula $\phi'$.
\end{lemma}

Recall the $3$ steps for proving the main theorem given in Introduction. The following theorem proves the first step.

\begin{theorem}
\label{thm_acd4ls}
 Let $\phi \in \mathcal{L_S}[<,+]$. Then there exists an active domain formula $\phi' \in \mathcal{L_S}[<,+]$ such that $\phi$ collapses to $\phi'$.
\end{theorem}
\begin{proof}
Let $\phi \in \mathcal{L_S}[<,+]$. We first claim that we can convert $\phi$ into a formula which uses only groups and $U_1$ as quantifiers. 
This follows from the Krohn-Rhodes decomposition theorem for monoids that every monoid can be decomposed into block products over groups and $U_1$. This decomposition can then be converted back into a formula using the groups and $U_1$ as quantifiers \cite{str_cirBook}. 

So without loss of generality we can assume $\phi$ has only group or $U_1$ quantifiers.
The proof is by induction on the quantifier depth. For the base case, let $\phi$ be a quantifier free formula. 
It is an active domain formula and therefore the claim holds. Let the claim be true for all formulas with quantifier depth $< d$. 
Lemma \ref{lem_acd4grpqnt} and Lemma \ref{lem_acd4u1} show that the claim is true for formulas of type $\phi = \mon{M}{m} z \langle \phi_1,\dots,\phi_K \rangle$ with 
quantifier depth $d$, when $M$ is a group or $U_1$ respectively. We are now left with proving that the claim is closed under conjunction and negation. So assume that 
formulas $\phi_1,\phi_2$ collapse to $\phi_1', \phi_2'$ respectively. That is there exist $\mathcal{R}_{\phi_1},\mathcal{R}_{\phi_2}\in \Nat$ such that $\mathcal{R}_{\phi_1}$ collapses
$\phi_1$ to $\phi_1'$ and $\mathcal{R}_{\phi_2}$ collapses $\phi_2$ to $\phi_2'$. Let $\mathcal{R}=max \{\mathcal{R}_{\phi_1},\mathcal{R}_{\phi_2}\}$. 
Then it is easy to see that $\mathcal{R}$ collapses $\phi_1 \wedge \phi_2$ to $\phi_1' \wedge \phi_2'$ and $\mathcal{R}_{\phi_1}$ collapses $\neg \phi_1$ to $\neg \phi_1'$.
\qed \end{proof}

We have shown above that all formulas in $\mathcal{L_S}[<,+]$ can be collapsed to active domain formulas. Now using a Ramsey type argument we obtain that addition is useless, giving us a formula in $\mathcal{L_S}[<]$. This corresponds to the second step in our three step proof strategy.

Let $R$ be any set of relations on $\Nat$ and let $\phi(x_1,\dots,x_t)$ be an active domain formula in $\mathcal{L_S}[R]$. We say that
$\phi$ has the \emph{Ramsey property} if for all infinite subsets $X$ of $\Nat$, there exists an infinite set $Y \subseteq X$ and an active
domain formula $\psi \in\mathcal{L_S}[<]$ that satisfies the following conditions. 
If $w \in \Sigma^*\lambda^{\omega}$ and $\nnp(w) \subseteq Y$, then for all $a_1,\dots,a_t \in Y$,
$$w \vDash \phi(a_1,\dots,a_t) \Leftrightarrow w \vDash \psi(a_1,\dots,a_t)$$ 

The Ramsey property for first order logic has been considered by Libkin \cite{libkin_FMT}. These results can be extended
to our logic.

\begin{theorem}
\label{thm_ramsey}
Let $R$ be a set of relations on $\mathbb{N}$. Every active domain formula in $\mathcal{L_S}[R]$ satisfies 
the Ramsey property.
\end{theorem}
\begin{proof}
  Let $\phi \in \mathcal{L_S}[R]$ be a formula. We now prove by induction on the structure of the formula.
  Let $P(x_1,\dots,x_k)$ be a term in $\phi$. We assume without loss of generality that for all $i \neq j, x_i \neq x_j$.
  Now consider the infinite complete hypergraph, whose vertices are labelled by numbers from $X$ and whose edges are $k$ tuple of vertices.
 Let $i_1,\dots,i_k$ be some permutation of numbers from $1$ to $k$. Consider the edge formed by the vertices $v_1<v_2<\dots<v_k$. 
 We color this edge by the formula $x_{i_1}<x_{i_2}<\dots<x_{i_k}$ if $P(v_{i_1},\dots,v_{i_k})$ is true. Observe that each edge can
 have multiple colors and therefore the total number of different colorings possible is $k!$. Ramsey theory gives us that there exists
 an infinite set $Y \subseteq X$, such that the induced subgraph on the vertices in $Y$ will have a monochromatic color, ie. all the
 edges will be colored using the same color. Let us assume that the edges in $Y$ are colored $x_1<x_2<\dots<x_k$. Then for all
 $a_1,\dots,a_t \in Y$
 $$a_1,\dots,a_t \models R(x_1,\dots,x_k) \Leftrightarrow a_1,\dots,a_t \models x_1<x_2<\dots<x_k$$
This shows that $P(x_1,\dots,x_k)$ satisfies the Ramsey property and thus all atomic formulas satisfy the Ramsey property.
We now show that Ramsey property is preserved while taking Boolean combination of formulas. Consider the formula $\phi_1(x_1,\dots,x_k)
\wedge \phi_2(x_1,\dots,x_k)$. We know that by induction hypothesis there exists a formula $\psi_1$ and an infinite set $X$ such that
for all $a_1,\dots,a_t \in X$, $w \models \phi_1(a_1,\dots,a_t) \Leftrightarrow w \models \psi(a_1,\dots,a_t)$. We can now find an
infinite set $Y \subseteq X$ and a formula $\psi_2$ such that the Ramsey property holds for the formula $\phi_2$. Therefore for all
$a_1,\dots,a_t \in Y$
$$w,a_1,\dots,a_k \vDash \phi_1 \wedge \phi_2 \Leftrightarrow w,a_1,\dots,a_k \vDash \psi_1 \wedge \psi_2$$ 
Similarly we can show that the Ramsey property holds for disjunctions and negations.
We need to now show that active domain quantification also preserves
Ramsey property. So let $X$ be an infinite subset of $\mathbb N$ and let
$$\phi(\vec x) = \mon{M}{m} z ~\neg \lambda(z) ~\langle \phi_1(z,\vec x),\dots,\phi_K(z,\vec x)\rangle$$
be a formula in $\mathcal{L_S}[R]$.
By induction hypothesis we know that there exists an infinite set $Y_1 \subseteq X$ and an active domain formula $\psi_1
\in \mathcal{L}[<]$ such that for all $\vec a \in Y_1^t$ the Ramsey property is satisfied. That is 
$w \models \phi_1(\vec a) \Leftrightarrow w \models \psi_1(\vec a)$. Now for $\phi_2$, using the infinite set $Y_1$ we can find an infinite 
set $Y_2 \subseteq Y_1$ and a formula $\psi_2$ satisfying the Ramsey property. Continuing like this will give us a set $Y_K$ and formulas
$\psi_1,\dots,\psi_K$ such that $\forall j\leq K$ and for all $w \in \Sigma^*\lambda^{\omega}$ with $\nnp(w) \subseteq Y_K$, we have
that $\forall b \in Y_K, \vec a \in Y_K^t,~ w \vDash \phi_j(b,\vec a) \Leftrightarrow w \vDash \psi_j(b, \vec a)$. Hence we also have 
that $\forall j \leq K$
$$\{b\in Y_K \mid w \vDash \phi_j(b,\vec a)\} = \{b\in Y_K \mid w \vDash \psi_j(b,\vec a)\}$$
Therefore for the formula $\psi = \mon{M}{m} z ~\neg \lambda(z) ~\langle \psi_1,\dots,\psi_K \rangle$, we
have $\forall w$ where $\nnp(w) \subseteq Y_K$ and $a_1,\dots,a_t \in Y_K$  that
$$w \vDash \phi(a_1,\dots,a_t) \Leftrightarrow w \vDash \psi(a_1,\dots,a_t)$$
Observe that $\psi$ is an active domain formula in $\mathcal{L_S}[<]$.
\qed \end{proof}

We continue with the third step of our three step proof strategy. 
\begin{lemma}
\label{lem_acdneutr}
 Every active domain sentence in $\mathcal{L_S}[<]$ define a language with a neutral letter.
\end{lemma}
\begin{proof}
Let $\phi \in \mathcal{L_S}[<]$ be an active domain formula over letter $\lambda \in \Sigma$. Let $w \in \Sigma^{\omega}$. 
Let $w' \in \Sigma^{\omega}$ got by inserting letter $\lambda$ in $w$ at some positions. Let $n_1 < n_2 < \dots$ belong to $\nnp(w)$ and 
$m_1< m_2 <\dots$ be in $\nnp(w')$. Let $\rho: \nnp(w) \rightarrow \nnp(w')$ be the bijective map $\rho(n_i) = m_i$. We show that for any subformula 
$\psi$ of $\phi$ and any $\vec t \in \nnp(w)^s$, we have that $w,\vec t \vDash \psi \Leftrightarrow w',\rho(\vec t) \vDash \psi$.
The claim holds for the atomic formula $x>y$, because $n_i > n_j$ iff $\rho(n_i) > \rho(n_j)$ for an $i,j$. Similarly the claim also
hold for all other atomic formulas $x<y, x=y$ and $a(x)$ for an $a \in \Sigma$. The claim remains to hold under
conjunctions, negations and active domain quantifications. Hence $w \models \phi \Leftrightarrow w' \models \phi$.
This proves that $\lambda$ is a neutral letter for $L(\phi)$.
\qed \end{proof}

Now we can prove our main theorem.

\begin{proof}[Proof of Theorem \ref{thm_cbconj}]
Let $\phi \in \mathcal{L_S}[<,+]$, such that $L(\phi)$ is a language with a neutral letter, $\lambda$. By Theorem \ref{thm_acd4ls} there exists an active domain sentence $\phi' \in \mathcal{L_S}[<,+]$ over $\lambda$ and a set $\mathcal{D}_\mathcal{R}$ such that $\mathcal{R}$ collapses 
$\phi$ to $\phi'$. Theorem \ref{thm_ramsey} now gives an active domain formula $\psi \in \mathcal{L_S}[<]$ and an infinite set $Y \subseteq \mathcal{D}_\mathcal{R}$. 
We now show that $L(\phi)=L(\psi)$.
Let $w \in \Sigma^*\lambda^{\omega}$. Consider the word $w' \in \Sigma^*\lambda^{\omega}$ got by inserting the neutral letter $\lambda$ in $w$ in such a way 
that $\nnp(w') \subseteq Y$. Since $L(\phi)$ is a language with a neutral letter we have that $w \models \phi \Leftrightarrow w' \vDash \phi$. From Theorem \ref{thm_acd4ls} and 
Theorem \ref{thm_ramsey} we get
$w' \vDash \phi \Leftrightarrow w' \vDash \phi' \Leftrightarrow w' \vDash \psi$.
Finally as shown in Lemma \ref{lem_acdneutr}, $\psi$ defines a language with a neutral letter and hence $w' \models \psi \Leftrightarrow
w \models \psi$. 
\qed \end{proof}


\bigskip
\section{Proof of Lemma \ref{lem_acd4grpqnt}} \label{sec_lemmaProof}
In this section we replace a group quantifier by an active domain formula. Here we make use of the fact that we can a priory restrict our domain as shown in the previous section.

Recall that $\phi = \mon{G}{m} z \langle \phi_1,\dots,\phi_K \rangle$ and $G=\{m_1,\dots, m_K,1\}$. 
We know that for all $i \leq K$, there exists $\mathcal{R}_{\phi_i}$ and a formula $\phi_i'$ such that $\mathcal{R}_{\phi_i}$ collapse $\phi_i$ to $\phi_i'$. 
Then clearly $max \{ \mathcal{R}_{\phi_i} \}$ collapse $\phi_i$ to $\phi_i'$ for all $i\leq K$.
So without loss of generality we assume $\phi_i$s are active domain formulas. 

Before we go in the details we will give a rough overview of the proof idea.
The group quantifier will evaluate a product $\prod_j u(j)$ where $u(j)$ is a group element that depends on the set of $i$ such that $w,j\models\phi_i$. So we start and analyze the sets $J_i=\{j\mid w,j\models\phi_i\}$. Since the formulas $\phi_i$ are active domain formulas, we will see that there exists a set of intervals such that inside an interval the set $J_i$ is periodic. Boundary points for these intervals are either points in the domain, or linear combinations of these. 
In the construction of the active domain formula for $\phi$ we will show how to iterate over all these boundary points in a strictly increasing order. An active domain quantifier can only iterate over active domain positions, hence we will need nested active domain quantifiers, and a way how to ``encode'' the boundary points by tuples of active domain positions in a unique and order preserving way. Additionally we need to deal with the periodic positions inside the intervals, without being able to compute the length of such an interval, or even check if the length is zero. Here will make use of the inverse elements that always exist in groups.
 
We start by analyzing the intervals which occur.
We will pick an $\mathcal{R}_{\phi} \geq max \{ \mathcal{R}_{\phi_i} \}$ to collapse the formula $\phi$. During the course of the proof we will require $\mathcal{R}_{\phi}$
 to be greater than a few others constants, which will be specified then. But always observe that $\mathcal{R}_{\phi}$ will depend only on $\phi$.

Since we consider a fixed set $\mathcal S$ for the rest of the paper, we will write $\mathcal{L}[<,+]$ for the logic $\mathcal{L_S}[<,+,0,\{\equiv_q:q>1\}]$.

\subsection{Intervals and Linear Functions} 
\label{linFunc}
We first show that every formula $\psi$ with at least one free variable has a normal form.
\begin{lemma}
 \label{lem_normalform}
Let $\psi(z) \in \mathcal{L}[<,+]$. Then there exists a formula $\hat \psi(z) \in \mathcal{L}[<,+]$ such that $\psi$ is equivalent to $\hat \psi$, where all 
atomic formulas in $\hat \psi$ with $z$ are of the form $z > \rho, z = \rho, z< \rho, z \equiv_n \rho$, where $\rho$ is a linear function on variables other than $z$.
 \end{lemma}

\begin{proof}
Terms in our logic are expressions of the form $$\alpha_0+\alpha_1x_1+\dots+\alpha_s x_s ~~~~~, \mbox{where} ~~\alpha_i \in \Nat$$ and atomic formulas are of the form
$$\sigma = \gamma, \sigma < \gamma, \sigma > \gamma, \sigma \equiv_m \gamma, c(\sigma)$$
whee $\sigma, \gamma$ are linear functions, $c \in \Sigma$ and $m > 1$. \\
Now using any $M \in \mathcal{S}$, where $m_1\in M$ is not the neutral element, we can rewrite $c(\sigma)$ as
$$\mon{M}{m_1} x ~\neg \lambda(x) \langle (   x = \sigma) \wedge c(x), false,\dots,false\rangle$$
Now consider the atomic formulas containing the free variable $z$ in $\psi(z)$. 
By multiplying with appropriate numbers, we can re-write these atomic formulas as $nz = \rho, nz < \rho, nz > \rho, nz \equiv_l \rho$ 
for one particular $n$, which is the least common multiple (lcm) of all the coefficients in $\psi$. Here $\rho$ does not contain $z$ and also it might contain subtraction. That is $nz=\rho$ might stand for $nz+\rho_1 = \rho_2$. Now we replace $nz$ by $z$ and 
conjunct the formula with $z \equiv_n 0$. 
\qed \end{proof}

For any formula $\psi(z)$, the notation $\hat \psi(z)$ denotes the normal form as in Lemma \ref{lem_normalform}.
Let $x_1,\dots,x_s$ be the bounded variables occurring in $\hat\phi_i(z)$ and $y_1,\dots,y_r$  be the free variables other than $z$ in $\hat\phi_i(z)$.
Hence the terms $\rho$ that appear in the formula $\hat\phi_i(z)$ can be identified as functions, $: ~\mathbb N^{s+r}\rightarrow\mathbb N$.

We collect all functions $\rho(\vec x,\vec y)$ that occur in the formulas $\hat\phi_i(z)$ for an $i\leq K$:  
$$R=\{\rho\mid \text{where $\rho$ is a linear term occurring in $\hat\phi_i(z)$}, i\leq K\}$$
We define the set $T$ of offsets as a set of terms which are functions using the variables $y_1,\dots,y_r$ as parameters:
$$T=\{\rho(0,\dots,0,y_1,\dots,y_r)\mid \rho\in R\}\cup\{0\}$$
Consider the set of absolute values of all the coefficients appearing in one of the functions in $R$. Let $\alpha' \in \Nat$ be the maximum value among these. That is 
$\alpha' = max \{ |\gamma| \mid f \in R, \gamma \mbox{ is a coefficient in } f\}$. Let $\Delta=s\cdot\alpha'$.
Now we can define our set of extended functions.
For a $t \in T$ we define a set of terms which are functions using the variables $x_1,\dots,x_s,y_1,\dots,y_r$ as parameters:
$$F_t=\big\{\sum_i^{s'} \alpha_i x_i + t\mid s'\leq s,-\Delta \leq \alpha_i\leq\Delta, \alpha_i \in \mathbb{N}\big\}.$$
We denote by $F=\cup_{t \in T} F_t$.

For a fixed word $w \in \Sigma^*\lambda^{\omega}$ and a fixed assignment of the free variables $\vec y$ to $\vec a$ we let 
$$B^{w,\vec a}=\{f(\vec d,\vec a)\mid t \in T, f\in F_t, \vec d \in\nnp(w)^{s'},d_1>d_2>\dots>d_{s'}\}$$
be the set of {\em boundary points}. Note that the assignments to the functions are of strictly decreasing order. Let $b_1<b_2<\ldots<b_l$ be the boundary points 
in $B^{w,\vec a}$. Then the following sets are called \emph{intervals}: $(-1,b_1),(b_1,b_2),\dots,(b_{l-1},b_l),(b_l,\infty)$. Here 
$(a,b) = \{x \in \Nat \mid a<x<b\}$.
We also split the set of points in $B^{w,\vec a}$ depending on the offset
$$B_t^{w,\vec a}=\{f(\vec d,\vec a)\mid f\in F_t,\vec d \in\nnp(w)^{s'},d_1>d_2>\dots>d_{s'}\}.$$


In the following Lemma we fix a word $w \in \Sigma^*\lambda^{\omega}$ and an $\vec a \in \Nat^r$.
\begin{lemma}\label{bplemma}
$\{\rho(d_1,\ldots,d_s,\vec a) \mid \rho \in R, d_i \in nnp(w)\} \cup \nnp(w) \subseteq B^{w,\vec a}$
\end{lemma}
\begin{proof}
Let $S=\{\rho(d_1,\ldots,d_s,\vec a) \mid \rho \in R, d_i \in nnp(w)\} \cup \nnp(w)$. 
Since $\rho(x_1)=x_1$ is in $F_t$, for some $t\in T$, we have $nnp(w) \subseteq B^{w,\vec a}$. Let $b\in S$. Then there is a function
$\rho=\sum_i^{s'} \alpha_i x_i+t(\vec y)$ in $F_t$ and values $p_1,\dots,p_{s'}\in nnp(w)$ such that $b=\rho(p_1,\dots,p_{s'},\vec a)$. Let $p_1'> p_2' >\dots>p_l'$ be the ordered set of all $p_i$s in the above assignment. We let $\rho'(x_1,\dots,x_l)=\sum_i \beta_{i} x_i+t$, where 
$\beta_i=\sum_{j:p_j=p'_i} \alpha_j$. Therefore $b=\rho'(p'_1,\dots,p'_l)$.
Since $|\beta_i|\leq \Delta\cdot s$ we have $\rho'\in F_t$ and hence $b\in B^{w,\vec a}_t$.
\qed \end{proof}

Let $q$ be the lcm of all $q'$ where $\equiv_{q'}$ occurs in one of the $\phi_i$. 
We need the following lemma, that inside an interval with only neutral letters, the congruence relations decide the truth of an active domain formula. 
\begin{lemma}
\label{lem_witInIntrv}
Let $a_1,\dots,a_r \in \Nat$ and let $c,d \in \Nat$ belong to the same interval in $B^{w, \vec a}$ such that $c \equiv_q d$. Then for all $i \leq K$:
$w,c \vDash \phi_i(z,\vec a) \Leftrightarrow w,d \vDash \phi_i(z, \vec a)$.
\end{lemma}
\begin{proof}
Proof is by induction on the structure of the formula $\hat \phi_i$. We will now show that $\forall b_i \in \nnp(w)$ and all
subformulas $\psi(z,\vec x,\vec y)$ of $\hat \phi_i$ that $w,c,\vec b, \vec a \vDash \psi \Leftrightarrow w,d,\vec b,\vec a \vDash
\psi$. The atomic formulas of $\hat \phi_i(z,\vec a)$ are of the following form: $z < \rho(\vec x,\vec a), z = \rho(\vec x,\vec a), z>
\rho(\vec x,\vec a) , z \equiv_{q'} \rho(\vec x,\vec a), a(z)$ and formulas which does not depend on $z$. It is clear that the truth of formulas which
does not depend on $z$, $a(z)$ and $z \equiv_{q'} \rho$ does not change whether we assign $c$ or $d$ to $z$. Let $\vec b \in nnp(w)^s$.
By Lemma \ref{bplemma} we know that $\rho(\vec b,\vec a)$ is in $B^{w,\vec a}$ and since $c,d$ lies in the same interval it follows that $c<
\rho(\vec b,\vec a) \Leftrightarrow d<(\vec b,\vec a)$. Similarly we can show that the truth of $z>\rho,z=\rho$ does not change on $z$
being assigned $c$ or $d$. Thus we have that the claim holds for atomic formulas.
The claim clearly holds for conjunction and negation of formulas. Now let the claim hold for subformulas $\psi_1,\dots,\psi_K$. Therefore $\forall i\leq K$ we have that
$\{\vec b \in \nnp(w)^s \mid w,c,\vec b,\vec a\vDash \psi_i\} = \{\vec b \in \nnp(w)^s \mid w,d,\vec b,\vec a \vDash \psi_i\}$.
Therefore we have that
\begin{eqnarray*}
w, c,b_2,\ldots,b_s,\vec a \vDash \mon{M}{m} x ~\neg \lambda(x) \langle \psi_1,\dots,\psi_K \rangle \\
\Leftrightarrow w, d,b_2,\ldots,b_s,\vec a \vDash \mon{M}{m} x ~\neg \lambda(x) \langle \psi_1,\dots,\psi_K \rangle
\end{eqnarray*}
And hence it is closed under active domain quantification.
\qed \end{proof}

The following Lemma deals with the infinite interval.

\begin{lemma}
 \label{lem_infInterval}
 Let $b$ belong to the infinite interval and $\vec a \in \Nat^r$. If $w, \vec a \vDash \phi$ then 
$w,b, \vec a \nvDash \phi_i$ for any $i\leq K$.
\end{lemma}
\begin{proof}
Let $i \leq K$ and $b$ be in the infinite interval and $w,b, \vec a \vDash \phi_i$. From Lemma \ref{lem_witInIntrv} we know that all points
$c \equiv_q b$ and such that $c$ is also in the infinite interval will be a witnesses for $\phi_i$. This means the set of witnesses is infinite and hence $w,\vec a \nvDash \phi$.
\qed \end{proof}

Lemma \ref{lem_witInIntrv} says that inside an interval, the congruence relations decide the satisfiability of the formulas $\phi_i$s.
This shows that it is enough to know the truth values of $\phi_i$ at a distance of $\geq q$ from the boundary points, since the truth
values inside an interval are going to repeat after every $q$ positions.
The rest of the proof demonstrates 
\begin{enumerate}
 \item How we can treat each $B_t$ differently. 
 \item There is an active domain formula which goes through the points in $B_t$ in an increasing order
\end{enumerate}
We fix the word $w \in \Sigma^*\lambda^{\omega}$ and assignment $\vec a$. Therefore we drop the superscripts in
$B^{w,\vec a}$ ($B_t^{w,\vec a}$) and call them $B$ ($B_t$). 
\subsection{Treating each $B_t$ differently}
Let $p=q|G|$, where $q$ was defined in the previous section and depends on the $\equiv_{q'}$ predicates.
For an element $g\in G$, we have $g^{|G|}=1_G$, so $g^x=g^{x+|G|}$.
Recall the definitions of $T, B$ from Section \ref{sec_lemmaProof}. 

Recall from the Preliminaries (Section \ref{sec_prelims}) that we denoted by $u(i)$ the group element at position $i$. That is
$u(i)=m_j$ iff $w,\vec a \models \phi_j \wedge \bigwedge_{l<j} \neg \phi_l$. Our aim is to give an active domain formula such that the formula evaluates to true iff the group element $\prod_{i=0} u(i)$ is equal to $m$. The rest of this subsection will be devoted to computing this product in a way which helps in building an active domain formula.

Let $b< b'$ be boundary points in $B$. Below we compute $\prod_{i=b+1}^{b'-1} u(i)$ in a different way:
$$\prod_{i=b+1}^{b'-1} u(i)=\prod_{i>b} u(i) \left( \prod_{i\geq b'} u(i)\right)^{-1}.$$
Observe that we can compute the product of the interval using two terms that both need to know only one boundary of the interval.
It becomes simpler if we note that the two products do not really need to multiply all the elements $u(i)$, for $i \geq b'$ but simply agree on a common set of elements to multiply. 

For a $b \in B$, we define the function $\IL(b)$ to be the length of the interval to the left of $b$. That is if $(b',b)$ form an interval then $\IL(b)= b- b'-1$. Similarly we define $\IR(b)$ to be the length of the interval to the right of $b$. For all $k \leq |T|$, we define functions $N_k(b)$ and $\hat N_k(b)$, which maps points $b\in B$ to a group element. 

$$N_0(b) = 
\begin{cases}
 u(b+1) u(b+2) \dots u(b+\IR(b)) \hfill  \text{if } \IR(b) < p & \\
 u(b+1) u(b+2) \dots u(b+r) & \\
 \hspace{2cm} \text{if } \IR(b)\geq p \text{ and } r < p, b+r \equiv_p 0 & \\
\end{cases}
$$
$$
\hat N_0(p)=
\begin{cases}
1_G \hspace{1.55cm} \text{if } \IL(b)< p & \\
u(b-p) \dots u(b-p+r) & \\
\hspace{2cm} \text{if } \IL(b)\geq p \text{ and } r < p, b+r \equiv_p 0 & \\
\end{cases}
$$

Inductively we define
$$N_k(b)=N_{k-1}(b)\prod_{\substack{b'\in B_{t_k}\\b'>b}} \left(\hat N_{k-1}(b')\right)^{-1} u(b') N_{k-1}(b'),$$
$$\hat N_k(b)=\hat N_{k-1}(b)\prod_{\substack{b'\in B_{t_k}\\b'>b}} \left(\hat N_{k-1}(b')\right)^{-1} u(b') N_{k-1}(b').$$

\noindent We first show how $N_k(b)$ and $N_k(b')$ are related for $b,b' \in B$.
\begin{lemma} Let $0 \leq k \leq |T|$.
Let $b<b'\in B$ such that there are no points $b'' \in \bigcup_{i>k} B_{t_i}$, where $b < b'' < b'$. Then
$N_k(b)(\hat N_k(b'))^{-1}=\prod_{i=b+1}^{b'-1} u(i)$.
\end{lemma}
\begin{proof}
We prove this by induction over $k$.
Let $k=0$ and let $(b,b')$ form an interval in $B$. If $b'-b \leq p$ then 
$$(N_0(b))(\hat N_0(b'))^{-1}=\big(u(b+1) u(b+2) \dots u(b+\IR(b))\big)(1_G)^{-1}=\prod_{i=b+1}^{b'-1} u(i)$$

If the interval is large, i.e. $b'-b>p$, then let $s,t \in \Nat$, be the smallest, resp. the largest numbers such that $b\leq s\leq
t\leq b'$ and $s\equiv_p t\equiv_p 0$.
Lemma \ref{lem_witInIntrv} shows that inside an interval all positions congruent modulo $q$ satisfy the same formulas. Therefore $u(b'-p) u(b'-p+1) \dots u(b'-1)=1_G$, and hence $(u(b'-p) u(b'-p+1) \dots u(t))^{-1}=(u(t+1)\dots u(b'-1))$. So 
$$N_0(b)(\hat N_0(b'))^{-1}=\big(u(b+1) u(b+2) \dots u(s)\big) \big(u(t+1)\dots u(b'-1)\big)=\prod_{i=b+1}^{b'-1} u(i)$$
The last equality being true since $u(s+1)\dots u(t)=1_G$

As induction hypothesis assume that the lemma is true for all $k'<k$.
Since for all $b''>b'$ the terms $\big(\hat N_{k-1}(b'')\big)^{-1} u(b'') N_{k-1}(b'')$ appear in both $N_k(b)$ and $\hat N_k(b')$ 
they cancel out (whatever they compute to). Thus 
$$N_k(b)(\hat N_k(b'))^{-1}=\bigg(N_{k-1}(b)\hspace*{-1mm}\prod_{\substack{b''\in B_{t_k}\\b<b''<b'}}\hspace*{-3mm}\big(\hat N_{k-1}(b'')\big)^{-1} u(b'') N_{k-1}(b'')\bigg)(\hat N_{k-1}(b'))^{-1}$$

Let $b=b_0<b_1<\dots<b_{x-1}<b_x=b'$ be all positions in $B_{t_k}$ between $b$ and $b'$. By the requirements of the lemma the only positions of $B$
between $b_i$ and $b_{i+1}$ are in $\bigcup_{i<k} B_{t_i}$. Writing out the product we get 
$$N_k(b)(\hat N_k(b'))^{-1} = N_{k-1}(b_0)\left (\hat N_{k-1}(b_1) \right)^{-1} \prod_{i=1}^{x-1} u(b_i) ~N_{k-1}(b_i)\left (\hat N_{k-1}(b_{i+1}) \right)^{-1}$$
By I.H.  $N_{k-1}(b_i) \left(\hat N_{k-1}(b_{i+1})\right)^{-1}=\prod_{i=b_i+1}^{b_{i+1}-1} u(i)$. Hence
$N_k(b)(\hat N_k(b'))=\prod_{i=b+1}^{b'-1} u(i).$
\qed \end{proof}

The following Lemma shows that $u(0)N_{|T|}(0)$ gives the product of the group elements.
\begin{lemma}\label{lem_whattocomp}
We have that $u(0) N_{|T|}(0)=\prod_i u(i)$.
\end{lemma}
\begin{proof}
Using appropriate induction hypothesis we get that $N_{|T|}(0)=\prod_{i=1}^{l} u(i)$, where $l>\max(B)$. The lemma now follows from Lemma \ref{lem_infInterval} which gives that $u(i)=1_G$ for every $i$ in the infinite interval.
\qed \end{proof}

We now give active domain formulas $\gamma^m$, $m \in G$, such that $\gamma^m$ is true iff $N_{|T|}(0)=m$. For this we make use of the
inductive definition of $N_k$ and show that there exists active domain formulas $\gamma^m$ such that $w \models \gamma^m(b) \Leftrightarrow N_k(b)=m$.
Similarly we give active domain formulas $\hat \gamma^m$ such that $w \models \hat \gamma^m(b) \Leftrightarrow \hat N_k(b)=m$.
Observe that $N_k(b)$ is got by computing the product of 
$\left(\hat N_{k-1}(b')\right)^{-1} u(b') N_{k-1}(b')$, over $b'$, where $b'$ strictly increases. This requires us to traverse the elements in $B_{t_{k-1}}$ in an increasing order. The following section builds a Sorting tree to sort the elements of $B_{t_{k-1}}$ in an increasing order.


\subsection{Sorting Tree} \label{subsec_sortingtree}
Let $t \in T$. The aim of this section is to create a data structure, which can traverse the elements in $B_t$ in an ascending order. 

%
%

For a $t \in T$, we define a tree called \emph{sorting tree}, $\mathcal{T}_t$ which corresponds to $B_t$. The tree satisfies the
following property.
If the leaves of the tree are enumerated from left to right, then we get the set $B_t$ in ascending order.
A node in $\mathcal{T}_t$ is labeled by a tuple $(f,A)$, where $f(x_1,\dots, x_l)$ is a function in $F_t$, $A$ an assignment for the variables in $f$ 
such that $A(x_1)>A(x_2)>\dots>A(x_l)$ and $\forall i\leq l: A(x_i) \in nnp(w)$. 

We show how to inductively built the tree. The root is labeled by the tuple $(t,\{\})$, where $t$ is the function which depends only on $\vec y$ (and hence constant on $\vec x$)
 and $\{\}$ is the empty assignment. The root is not marked a leaf node.

Consider the internal node $(f(x_1,\dots,x_l),A)$. It will have three kinds of children ordered from left to right as follows.
\begin{enumerate}
 \item Left children: These are labeled by tuples of the form $(f_{\alpha}',A'_j)$ 
where\linebreak $f'_{\alpha}(x_1,\dots,x_{l+1})=f(x_1,\dots,x_l)+\alpha x_{l+1}$ and $-\Delta \leq \alpha < 0$, $-\alpha \in \Nat$, $A'_j=A\cup[x_{l+1}\mapsto j]$, where $j < A(x_l)$ and $j \in\nnp(w)$.

 The tuples $(f'_{\alpha_1},A'_{j_1})$ is on the left of $(f'_{\alpha_2},A'_{j _2})$ if $j_1>j_2$ or if $j_1=j_2$ and $\alpha_1<\alpha_2$. 
 \item Middle child: It is labeled by the tuple $(f'',A)$ where $f''(x_1,\dots,x_{l})=f(x_1,\dots,x_l)$. It is marked a {\bf leaf} node.
  \item Right children: These are labeled by tuples of the form $(f_{\alpha}',A'_j)$ 
where $f'_{\alpha}(x_1,\dots,x_{l+1})=f(x_1,\dots,x_l)+\alpha x_{l+1}$ and $0< \alpha \leq \Delta$, $\alpha \in \Nat$, $A'_j=A\cup[x_{l+1}\mapsto j]$, where $j < A(x_l)$ and $j \in\nnp(w)$.

The tuple $(f'_{\alpha_1},A'_{j_1})$ is on the left of $(f'_{\alpha_2},A'_{j_ 2})$ if $j_1<j_2$ or $j_1=j_2$ and $\alpha_1<\alpha_2$. 
\end{enumerate}

Observe that if there is no $j$ such that $j < A(x_l)$ and $j \in\nnp(w)$, then $(f,A)$ will only have the child $(f'',A)$.

Note that in our tree construction the values of the children of a node increase from left to right. The tree is built until
all functions with $s$ variables appear in leaves and hence the depth of the tree is $s+2$.
Figure \ref{fig_sortingTree} shows part of a tree, where $\Delta=2$, $t=0$, $\mathcal{R}=5$ and $\nnp(w) = \{5,25,625\} \subseteq
\mathcal{D}_{\mathcal{R}}$.

\begin{figure*}
\centering
\vspace*{-4mm}
\includegraphics[width=0.9\textwidth]{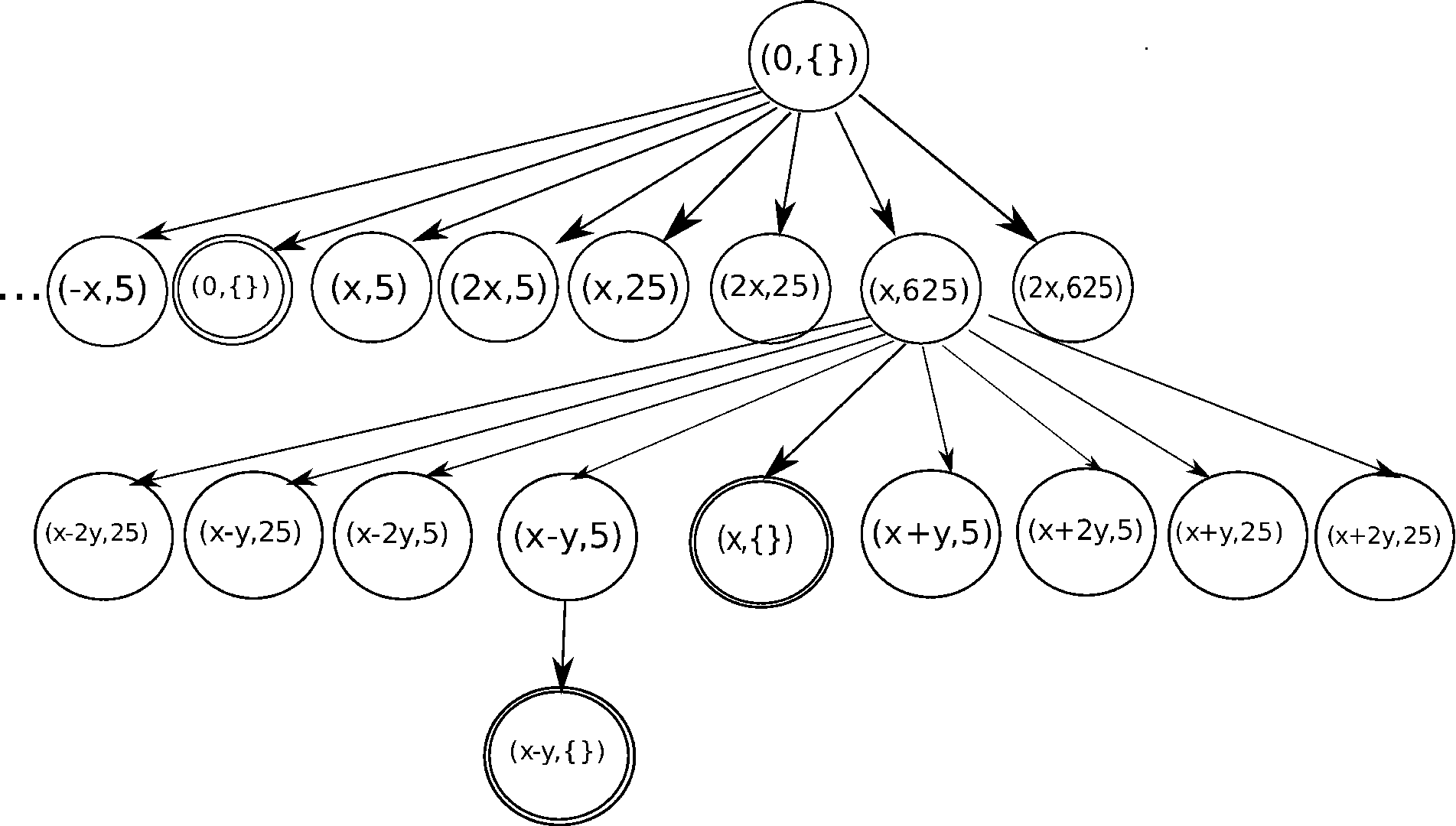} 
\vspace*{-4mm}
\caption{Sorting Tree: The double circles represent leaves of the tree. The nodes of the tree are labelled $(f,A)$, 
where $A$ is an assignment for the function $f$ and $t=0$. For better presentation we only show the assignment to the newly introduced variable in a node. 
For example, the tuple $(x-2y,25)$ assigns $x = 625$ and $y=25$. The assignment to $x$ is given in the node's parent.}
\label{fig_sortingTree}
\vspace*{-5mm}
\end{figure*} 

The following lemma holds if $\mathcal{R} > 3s\Delta$. We also assume that $nnp(w) \subseteq \mathcal{D}_{\mathcal{R}}$.
Given a node $(f,A)$, we say the value of the node is the function $f$ evaluated under the assignment of $A$ (denoted by $f(A)$).
\begin{lemma}
\label{lem_child}
Let $N$ be an internal node labeled by a function $f(x_1,\dots,x_l)$ with $l<s$ and an assignment $A$.
If $A(x_l) = \mathcal{R}^c$ for some $c\geq1$, then the children of this node have values in the range $[f(A)-\Delta \mathcal{R}^{c-1},
f(A)+\Delta \mathcal{R}^{c-1}]$.
Moreover the values of the children increases from left to right.
\end{lemma}
\begin{proof}
By construction.
\qed \end{proof}

Next we show that for any two neighboring nodes in the tree, the values in the leaves of the subtree rooted at the left node is less 
than the values in the leaves of the subtree rooted at the right node. Let $V_{(f,A)}$ denote the set of values in the leaves of the subtree rooted at $(f,A)$.

\begin{lemma}
 \label{lem_neighNoIntersect}
Let $(f,A)$ and $(f',A')$ be neighboring nodes of the same parent such that $(f,A)$ is to the left of $(f',A')$. Then $u<v$ for every $u \in V_{(f,A)}$ and $v \in V_{(f',A')}$.
\end{lemma}
\begin{proof}
Let $f=\sum_{i=1}^{l-1} \alpha_i x_i+\alpha_lx_l+t$ and $f' = \sum_{i=1}^{l-1} \alpha_i x_i+\alpha'_lx_l+t$. We show that the rightmost element, $u$ in $V_{(f,A)}$ is less 
than the left most element, $v$ in $V_{(f',A')}$. From Lemma \ref{lem_child} and applying induction on the depth of the tree, one can show that $u \leq
f(A)+(s-l)\Delta \mathcal{R}^{c-1}$ and 
$v \geq f'(A')-(s-l) \Delta \mathcal{R}^{c'-1}$. Here $\mathcal{R}^c,\mathcal{R}^{c'}$ are the minimum assignments in $A$ and $A'$ respectively.
Let us assume that both coefficients $\alpha_l',\alpha_l > 0$. A similar analysis can be given for other combinations of $\alpha_{l}'$ and
$\alpha_l$. 
Now since $(f,A)$ is the left neighbor of $(f',A')$ we have $\mathcal{R}^{c}<\mathcal{R}^{c'}$.  Then 
$v-u \geq \alpha'_l\mathcal{R}^{c'}-(s-l) \Delta \mathcal{R}^{c'-1}-\alpha_l\mathcal{R}^c-(s-l)\Delta \mathcal{R}^{c-1} \geq \mathcal{R}^{c'}
-3s\Delta \mathcal{R}^{c'-1} > 0$. The claim follows, since $\mathcal{R} > 3s\Delta$.
\qed \end{proof}
The next lemma says that the values of the leaves of the tree increases as we traverse from left to right.

\begin{lemma} \label{lem_tree}
Let $(f,A)$ and $(f',A')$ be two distinct nodes such that $f(A)<f'(A')$. Then $(f,A)$ appear to the left of $(f',A')$.
\end{lemma}
\begin{proof}
This follows from Lemma \ref{lem_neighNoIntersect}.
\qed \end{proof}


\begin{lemma}[Tree Lemma]\label{lem_treelemma}
Fix $t\in T$. Assume that for every $b\in B_{t}$ we have an element $g_b\in G$.
For all $f\in F_t$, $m\in G$ let $\gamma_f^m(\vec x)$ be active domain formulas  such that $w,\vec d\models\gamma_f^m(\vec x)\text{ iff }g_{f(\vec d)}=m.$
Then there are active domain formulas $\Gamma^{m'}$ such that $w\models\Gamma^{m'}\text{ iff }\prod_{b\in B_t} g_b=m'.$
\end{lemma}
\begin{proof}
We will use the sorting tree, $\mathcal{T}_t$ corresponding to $B_t$ for the construction of our formula. 
Recall that the nodes are labeled by tuples $(f,A)$, where $f$ is a function and $A$ is the assignment of the parameters of $f$. 
Let $V_{(f,A)}\subseteq B_t$ be the set of values at the leaves of the subtree rooted at the node labeled by $(f,A)$, and $g_{(f,A)}=\prod_{b\in V_{(f,A)}} g_b$.
We will do induction on the depth $D$ of the tree.
Let $\tau^{m,D}_{f}(\vec x)$ be a formula such that 
$w,\vec d \models \tau^{m,D}_{f}(\vec x)~~\text{$~~$iff }\prod_{b\in V_{(f,\vec d)}} g_b=m$ where $(f, \vec d)$ is the label of a node that has a subtree of depth at most $D$.
Hence we multiply all group elements $g_b$ for which $b$ is in $V_{(f,\vec d)}$.

{\bf Base Case (leaves):} 
We define $\tau^{m,0}_{f}(\vec x)=\gamma_f^m(\vec x)$.

{\bf Induction Step:} Let us assume that the claim is true for all nodes with a subtree of depth at most $D$.
Let the node labeled by $(f,A)$ have a subtree of depth $D+1$. We will need to specify the formula $\tau^{m,D+1}_{f}(\vec x)$, where $\vec x$ agrees with the assignment $A$.
For every child $(f',A')$ of $(f,A)$ the depth of the corresponding subtree is less than or equal to $D$. Hence we know we have already formulas by induction.

Recall what the children of $(f,A)$ are: They are of form $(f'_{\alpha},A'_j)$ and $(f'',A_j)$. Moreover all nodes
$(f'_{\alpha},A'_j)$, where $\alpha$ is
negative, come to the left of $(f'',A_j)$ and all nodes $(f'_{\alpha},A'_j)$, where $\alpha$ is positive, come to its right.

We start by grouping some of the children and computing their product.
We let $T^-(A'_j)$ be the product of all subtrees labeled by $(f'_{\alpha},A'_j)$ for $\alpha=-\Delta,-\Delta+1,\dots,-1$.
This is a finite product so we can compute this by a Boolean combination of the formulas $\tau^{m,D}_{f'_{\alpha}}(\vec x,x_{l+1})$.
$$\pi^{-,m,D}_f(\vec x,x_{l+1}) \defs \bigvee_{m_{-\Delta}\dots m_{-1}=m} \bigg(\bigwedge_{\alpha=-\Delta}^{-1} \tau^{m_{\alpha},D}_{f'_{\alpha}}(\vec x,x_{l+1})\bigg)$$
Now we want to compute the product $\left(\prod_{j\in\nnp(w)} (T^-(A'_j))^{-1}\right)^{-1}$ which is the product of the $T^-(A'_j)$ where $j\in\nnp(w)$ is decreasing.
But this can be computed using an active domain group quantifier, $\tau_f^{-,m,D}(\vec x)$ as follows:
$$\tau_f^{-,m,D}(\vec x) = \mon{G}{m^{-1}} x_{l+1} ~\big(\neg \lambda(x_{l+1}) \wedge (x_l > x_{l+1})\big) \hspace*{1.5cm}$$
$$\hspace*{2cm}\langle \pi_f^{-,m_1^{-1},D}(\vec x,x_{l+1}),\dots, \pi_f^{-,m_K^{-1},D}(\vec x,x_{l+1})\rangle$$ 
Recall that the elements of group $G$ are ordered $m_1,\ldots,m_K$.
For the single node $(f'',A)$ we already have the formulas $\tau^{m,D}_{f''}(\vec x)$ by induction (here we have $\vec x$ since the assignment $A$ is the same for $(f,A)$ and $(f'',A)$).

Similarly we define formulas $\pi^{+,m,D}_f(\vec x,x_{l+1})$ for the positive coefficients, and compute their product $\prod_{j\in\nnp(w)} T^+(A'_j)$ in an increasing order.
$$\pi^{+,m,D}_f(\vec x,x_{l+1}) \defs \bigvee_{m_{1}\dots m_{\Delta}=m} \bigg(\bigwedge_{\alpha=1}^{\Delta}
\tau^{m_{\alpha},D}_{f'_{\alpha}}(\vec x,x_{l+1})\bigg)$$
$$\tau^{+,m,D}_{f}(\vec x) \defs \mon{G}{m} x_{l+1} \big(\neg \lambda(x_{l+1}) \wedge (x_l > x_{l+1})\big)\hspace*{1.5cm}$$
$$\hspace*{2cm}\langle \pi_f^{+,m_1,D}(\vec x,x_{l+1}),\dots, \pi_f^{+,m_K,D}(\vec x,x_{l+1})\rangle$$
We have now computed the product of the group elements for the three different groups of children. So by a Boolean combination over
these formulas we get $\tau^{m,D+1}_{f}(\vec x)$:
$$\bigvee_{m'm''m'''=m} \big(\tau^{-,m',D}_{f}(\vec x) \wedge \tau^{m'',D}_{f''}(\vec x) \wedge \tau^{+,m''',D}_{f}(\vec x)\big)$$
So finally we get $\Gamma^{m'}$ which is same as the formula $\tau_ {(0,\{\})}^{m',s+2}$, which is valid at the root of the tree.
\qed \end{proof}

Since the above lemma holds only for $\mathcal{R}>3s\Delta$, our $\mathcal{R}_{\phi}$ should be greater than $3s \Delta$.
 
\subsection{Constructing the active domain formula} 
We know that for every $b\in B$ there is a function $f\in F$, $d_1,\dots,d_{s'}\in\nnp(w)$, such that $b=f(\vec d,\vec a)$, where $\vec
a$ is the fixed assignment to the variables $\vec y$.
We will use this encoding of a position and define a formula $\nu^m_{k,f}$ such that 
$$w,\vec d, \vec a\models \nu^m_{k,f}(\vec x,\vec y) \Leftrightarrow N_k(f(\vec d,\vec a))=m$$
Similarly we define formulas $\hat \nu^m_{k,f}$ such that $w, \vec d, \vec a \models \hat \nu^m_{k,f}(\vec
x, \vec y)$ iff $\hat N_k(f(\vec d, \vec a))=m$.

We show this by induction over $k\leq |T|$. Starting with the base case $k=0$.

\begin{lemma}\label{lem_formbase}
Let $\vec a \in \Nat^r$.  For each $m \in G$, there is an active domain formula $\nu^m_{0,f}(\vec x,\vec y)$ in $\mathcal L[<,+]$, 
such that if $w \models\nu^m_{0,f}(\vec d,\vec a)$ then $N_0(f(\vec d,\vec a))=m$. \\
Similarly there is an active domain formula $\hat \nu^m_{0,f}(\vec x,\vec y)$ in $\mathcal L[<,+]$ such that if  
$w,\vec d\models\hat \nu^m_{0,f}(\vec x,\vec a)$ then $\hat N_0(f(\vec d,\vec a))=m$.
\end{lemma}
\begin{proof}
  For an $i \leq K$, we denote by $\tilde \phi_{m_i}$ the formula $\bigwedge_{j<i} \neg \phi_j(x,\vec y) \wedge \phi_i(x,\vec y)$.
For a $l \in \Nat$, the following formula checks if there is a point $b'$ in $B$ such that $b+l=b'$. Since in each $B$ there is at most
one such element, we can use the group quantifier to simulate the existential quantifier.
$$\delta^l_{f} \defs \bigvee_{f'\in F \backslash f} \mon{G}{m_1} \vec x' \langle f'(\vec x',\vec y)=f(\vec x,\vec y)+l,false, \dots, false \rangle$$
So we have that $\IR(b)=l$ iff $\delta^{l+1}_f \wedge \bigwedge_{l'<l} \neg \delta^l_f$ is true. We define $\pi^{m,l}_f$ to be true if the product of the first $l$ group elements is $m$. 
$$\pi^{m,l}_f \defs \bigvee_{g_0\dots g_{l}=m}\bigg( \bigwedge_{i=0}^{l} \tilde \phi_{g_i}(f(\vec x,\vec y)+i)\bigg)$$
Now we have two cases to consider. \\
{\bf Case $\IR(b) < p$:} 
For each of the case $b<b'$ such that $l=b'-b \leq p$, the formula $\pi^{m,l}_f$ compute the product of the group elements. 
Hence $\nu^m_{0,f}$ in this case can be given as:
$$\bigwedge_{l=0}^{p-1} \bigg(\big(\delta^l_f(\vec x,\vec y) \wedge \bigwedge_{l'=0}^{l-1} \neg \delta^{l'}_f(\vec x,\vec y)\big) \rightarrow
\pi^{m,l}_f(\vec x,\vec y)\bigg)$$
{\bf Case $\IR(b) \geq p$:} 
When $b'-b > p$ we have to compute the product for the first $r$ group elements, where $b + r \equiv_p 0$ and $r<p$. Therefore
$\nu^m_{0,f}$ in this case is 
$$\bigwedge_{l=0}^{p-1} \big( f(\vec x,\vec y)+r \equiv_p 0 \big)\rightarrow \pi^{m,r}_f(\vec x,\vec y)$$
A Boolean combination over $\delta^l_f$ can differentiate the two cases.
Similarly we can give active domain formulas $\hat \nu^m_{0,f}(\vec x,\vec y)$.
\qed \end{proof}

The induction step follows.
\begin{lemma}\label{lem_formstep}
  Let $\vec a \in \Nat^r$. For each $m \in G$, there is an active domain formula $\nu^m_{k,f}$ in $\mathcal L[<,+]$,
such that $w,\vec d, \vec a \models\nu^m_{k,f}(\vec x,\vec y)$ then $N_k(f(\vec d,\vec a))=m$. \\
Similarly there is an active domain formula $\hat \nu^m_{k,f}(\vec x,\vec y)$ in $\mathcal L[<,+]$,
such that $w,\vec d, \vec a \models\hat \nu^m_{k,f}(\vec x,\vec y)$ then
$\hat N_k(f(\vec d, \vec a))=m$.
\end{lemma}
\begin{proof}
  For all $m \in G$ and $f' \in F_{t_k}$ we give formulas $\gamma_{f'}^m$ such that for all $\vec d' \in nnp(w)^s$ the
  following holds. Let $f'(\vec d', \vec a)=b'$ and $f(\vec d,\vec a)=b$. Then 
  $w, \vec d, \vec d',\vec a \models \gamma_{f'}^m(\vec x, \vec z,\vec y) \Leftrightarrow b' \leq b$  and  $m = 1_G$
  or  $b' > b$  and $\left(\hat N_{k-1}(b')\right)^{-1} u(b') N_{k-1}(b') = m$.
  By induction hypothesis there exists formulas $\nu^m_{k-1,f'}$ and $\hat \nu^m_{k-1,f'}$ which corresponds to $N_{k-1}(f'(\vec x,\vec y))$
  and $\hat N_{k-1}(f'(\vec x,\vec y))$ respectively. Taking a Boolean combination over these formulas we get the required formula $\gamma_{f'}^m$.
  We now apply our Tree Lemma \ref{lem_treelemma} which gives us formulas $\Gamma^m$, 
  for all $m \in G$, such that 
  $$w,\vec d,\vec a \models \Gamma^{m}(\vec x,\vec y) \Leftrightarrow
  w \models \prod_{\substack{b'\in B_{t_k}\\b'>b}} \left(\hat N_{k-1}(b')\right)^{-1} u(b') N_{k-1}(b') = m$$ 
  Taking Boolean combination over $\Gamma^m$ and $\hat \nu^m_{k-1,f}$ will give us the formula $\nu^m_{k,f}$.
Similarly we can build active domain formulas $\hat \nu^m_{k,f}(\vec x,\vec y)$, for all $m \in G$.
\qed \end{proof}

\begin{proof}[Proof of Lemma \ref{lem_acd4grpqnt}] 
By Lemma \ref{lem_whattocomp} we know that it suffices to compute $N_{|T|}(0)$ and by Lemma \ref{lem_formstep} we know that there are
active domain formulas $\Gamma^m \in \mathcal L[<,+]$ such that $N_{|T|}(0) = m$ iff $w,\vec 0,\vec a \models \Gamma^m$

We need to do one last thing. Check that the infinite interval evaluates to $1_G$. Replace all formulas $z>\rho$, $z<\rho$, $c(z)$ for
a $c \neq \lambda$ and $\lambda(z)$ by $true$, $false$, $false$, $true$ respectively in the formulas $\hat \phi_i$ and call these formulas $\hat\psi_i$. 
There exists a witness in the infinite interval for the formula $\hat \phi_i$ iff $\hat \psi_i$ evaluates to true. 
By Theorem \ref{lem_infInterval} there should not be any witness in the infinite interval. Hence there exists a $\hat \psi_i$ which
evaluates to true iff the infinite interval does not evaluate to $1_G$.
\qed \end{proof}

\bigskip

\bigskip
\section{Discussion}
We have shown that in the presence of a neutral letter the addition relation collapse to linear ordering no matter what monoid
quantifier is been used. All languages definable using monoid quantifiers and an order predicate, on the other hand, are regular \cite{barr_uniformNC1}. 
Now using semigroup theoretic methods we can separate these classes \cite{str_cirBook}. This enabled us to show separation between
various logics which uses addition and order predicates. 

Unfortunately if both addition and multiplication are present, then the collapse does not happen. It is also interesting to note that non-solvable
groups do not show any surprising property if only addition is present, but as we know from Barrington's theorem non-solvable groups
behave quite differently when both addition and multiplication are present.

The ultimate objective is to show non-expressibility results for arbitrary predicates or at least when both addition and multiplication
are present. As a first step one can look at extending these results for other kinds of predicates. 

Another way to look at separating the ``natural uniform'' versions of the complexity classes will be to ask whether one can come up
with other suitable restrictions on the set of languages. Inside this restricted set of languages can one show addition and
multiplication collapse to order relation? This seems to be the idea Straubing considers in \cite{str_inexpRegLan}. Straubing
\cite{str_cirBook} proposes word problems over Regular language as a suitable restriction, while McKenzie, Thomas, Vollmer \cite{mckenzie_extUniformity} consider context free languages as a restriction.

Another interesting question which our result fails to answer is whether word problems over non-solvable groups can be defined in
$\MAJ[<,+]$ \cite{krebs_infGroups}?
\section*{Acknowledgement}
\noindent We like to thank Baskar Anguraj, Christoph Behle, Micha\"el Cadilhac, Klaus-J\"orn Lange, Nutan Limaye, T. Mubeena, and Ramanujam for a lot of useful comments on the draft of this paper.

\end{document}